\documentclass[12pt]{article}
\usepackage{babel}
\usepackage[T1]{fontenc}
\usepackage[latin9]{inputenc}
\usepackage{verbatim}
\usepackage{amsmath}
\usepackage{amsthm}
\usepackage{amssymb}
\usepackage{graphicx}
\usepackage[a4paper,margin=1in]{geometry}
\usepackage{caption}
\setcaptionwidth{0.8\textwidth}

\usepackage{xcolor}

\theoremstyle{plain}
\newtheorem{thm}{\protect\theoremname}
\theoremstyle{remark}

\theoremstyle{remark}

\theoremstyle{definition}

\theoremstyle{plain}

\theoremstyle{plain}

\theoremstyle{plain}

\providecommand{\casename}{Case}
\providecommand{\claimname}{Claim}
\providecommand{\corollaryname}{Corollary}
\providecommand{\definitionname}{Definition}
\providecommand{\lemmaname}{Lemma}
\providecommand{\propositionname}{Proposition}
\providecommand{\remarkname}{Remark}
\providecommand{\theoremname}{Theorem}

\usepackage{physics}
\usepackage{bbm}
\usepackage{enumitem}
\usepackage{amsfonts}

\usepackage{authblk} %for affiliations to appear before the abstract
\newlist{casenv}{enumerate}{4}
\setlist[casenv]{leftmargin=*,align=left,widest={iiii}}
\setlist[casenv,1]{label={{\itshape\ \casename} \arabic*.},ref=\arabic*}
\setlist[casenv,2]{label={{\itshape\ \casename} \roman*.},ref=\roman*}
\setlist[casenv,3]{label={{\itshape\ \casename\ \alph*.}},ref=\alph*}
\setlist[casenv,4]{label={{\itshape\ \casename} \arabic*.},ref=\arabic*}

\makeatother

\begin{document}

\title{Observers of quantum systems cannot agree to disagree}

\author[1]{Patricia Contreras-Tejada}

\author[2]{Giannicola Scarpa}

\author[3]{Aleksander M. Kubicki}

\author[4]{Adam Brandenburger}

\author[5]{Pierfrancesco La Mura}

\affil[1]{Instituto de Ciencias Matem\'aticas, 28049 Madrid, Spain,
patrcont@ucm.es}
\affil[2]{Escuela T\'ecnica Superior de Ingenier\'ia de Sistemas Inform\'aticos, Universidad Polit\'ecnica de Madrid, 28031 Madrid, Spain, g.scarpa@upm.es (Corresponding Author)}
\affil[3]{Departamento de An\'alisis Matem\'atico y Matem\'atica Aplicada, Universidad Complutense de Madrid, 28040 Madrid, Spain, amkubickif@gmail.com}
\affil[4]{Stern School of Business, Tandon School of Engineering, NYU Shanghai,
New York University, New York, NY 10012, U.S.A., adam.brandenburger@nyu.edu,
http://www.adambrandenburger.com}
\affil[5]{HHL Leipzig Graduate School of Management, 04109 Leipzig, Germany, plamura@hhl.de}

\date{}
\maketitle
\thispagestyle{empty}

\abstract{
Is the world quantum? An active research line in quantum foundations is devoted to exploring what constraints can rule out the postquantum theories that are consistent with experimentally observed results. We explore this question in the context of epistemics, and ask whether agreement between observers can serve as a physical principle that must hold for any theory of the world. Aumann's seminal Agreement Theorem states that two observers (of classical systems) cannot agree to disagree. We propose an extension of this theorem to no-signaling settings. In particular, we establish an Agreement Theorem for observers of quantum systems, while we construct examples of (postquantum) no-signaling boxes where observers can agree to disagree. The PR box is an extremal instance of this phenomenon. These results make it plausible that agreement between observers might be a physical principle, while they also establish links between the fields of epistemics and quantum information that seem worthy of further exploration.
}

% Table of contents: uncomment to show
%
%\tableofcontents{}\bigskip{}

\newpage

\section{Introduction}

Quantum mechanics famously made its creators uncomfortable. 
It is highly counterintuitive and, almost a century after its introduction, it still sparks much conceptual and philosophical discussion. Indeed, an active line of research in quantum foundations deals with the problem of singling out quantum theory from other post-classical physical theories. This field is a delicate balance between proposals for new theories that are `tidier' than quantum mechanics \cite{spekkens_evidence_2007,larsson_contextual_2012} and proposals for desirable physical principles that such theories should obey \cite{popescu_quantum_1994,clifton_characterizing_2003,pawlowski_information_2009,sun_no_2019,yan_quantum_2013}.

In the domain of classical probability theory, Aumann proved that Bayesian agents cannot agree to disagree \cite{aumann_agreeing_1976}. A slightly more general restatement of Aumann's theorem, which we will refer to as the classical agreement theorem, states that, if Alice and Bob, based on their partial information, assign probabilities $q_{A},\,q_{B}$, respectively, to perfectly correlated events, and these probabilities are common certainty between them, then $q_{A}=q_{B}$. ``Certainty'' means assigning probability 1, and ``common certainty'' means that Alice is certain about $q_{B}$; Bob is certain about $q_{A}$; Alice is certain about Bob being certain about $q_{A}$; Bob is certain about Alice being certain about $q_{A}$; and so on infinitely.

This result is considered a basic requirement in classical epistemics, which is the formal study of  the knowledge and beliefs of the agents in a system. The classical agreement theorem has been used to show that two risk-neutral agents, starting from a common prior, cannot agree to bet with each other \cite{geanakoplos_sebenius_1983}, to prove ``no-trade'' theorems for efficient markets \cite{milgrom_stokey_1982}, and to establish epistemic conditions for Nash equilibrium \cite{aumann_brandenburger_1995}. 

Focusing on the quantum domain, a fundamental result of quantum mechanics is that no local hidden-variable theory can model the results of all quantum experiments \cite{bell_einstein_1964}. This implies that the classical Bayesian model does not apply, so the classical agreement theorem need not hold. The question then arises: Can observers of quantum mechanical phenomena agree to disagree?

In this work, we answer the above question in the negative. We define two notions of disagreement inspired by Aumann's theorem. One is a direct analogue to the classical agreement theorem, and the other one relaxes the common certainty condition while requiring that the probability estimates differ maximally. We find that neither kind of disagreement occurs for classical or quantum systems.  However, both kinds of disagreement do occur in postquantum environments. In fact, we characterize no-signaling distributions displaying these behaviours. 
We then put our two characterizations together and search for distributions that satisfy both notions of disagreement: We find that the PR box \cite{popescu_quantum_1994} is of this kind---i.e., it displays extremal disagreement in the above sense. Since the PR box is also an extreme instance of a no-signaling box as a non-local resource  \cite{popescu_quantum_1994,barret_nonlocal_2005}, our findings suggest a deeper relation between the quantification of disagreement and the quantification of non-locality.

If a physical theory were to allow agents to agree to disagree, then undesirable consequences in the settings of Refs. \cite{geanakoplos_sebenius_1983,milgrom_stokey_1982,aumann_brandenburger_1995} could happen. This is why the impossibility of agreeing to disagree is a desirable feature for all physical theories, and why we propose that it should be elevated to a physical principle. Its simplicity makes it convenient for testing the consistency of new postquantum theories.

\section{Results}

\subsection{Classical agreement theorem}\label{sec:classical-agreement}

We start with an intuition about the setup behind the classical agreement theorem. Suppose that Alice and Bob share a classical system, which can, thus, be described by a local hidden-variable model (in Aumann's language, each value of the variable represents a state of the world). But the observers do not know which value of the hidden variable is the one that holds (i.e., which is the true state of the world). Instead, each observer can perform only one local measurement on the system. Each measurement corresponds to a partition over the values of the variable, and the result reveals which partition element contains the value that holds. The probability of each outcome is the sum of the probabilities of the values in the corresponding partition element.

Suppose, also, that Alice is interested in estimating the probability of an event (i.e., a set of values of the variable) that does not correspond to an element in her partition. Then, she can  calculate only the conditional probability of the event given the outcome of her local measurement by Bayesian inference. The same applies for Bob. Suppose that Alice and Bob are interested in events that are perfectly correlated (i.e., with probability 1, either both happen or neither happens). Then, the classical agreement theorem says that, if their estimates are common certainty, they must be equal.

Common certainty means that Alice is certain of (i.e., assigns probability 1 to) Bob's estimate, Bob is certain of Alice's estimate, Alice is certain of Bob's certainty of Alice's estimate, and so on.

When formalizing these notions, we refer to a probability space, together with some given partitions, as a (classical) ontological model. Ontological models appearing in the literature (see, e.g., \cite{ferrie_quasi-probability_2011}) also contain a set of preparations underlying the distribution over the probability space, and the partitions are usually phrased in terms of measurements and outcomes. However, we consider preparations implicit and use the language of partitions to bridge the gap between classical probability spaces and no-signaling boxes more smoothly.

For the sake of simplicity and following Aumann, we restrict our analysis to two observers, Alice and Bob. Aumann's original theorem considers common knowledge of one single event of interest to both observers. We provide a slight generalization with common certainty about two perfectly correlated events of interest, one for each observer. This allows us to move to the framework of no-signaling boxes that we will use later. This is what we call the classical agreement theorem. This terminology can be further motivated by the fact that, for purely classical situations, both statements---the original Aumann's theorem and our formulation with perfectly correlated events---can be proven to be equivalent (as long as states of the world with null probability are ignored, as in Ref. \cite{aaronson_complexity_2005}).

Consider a probability space $(\Omega,\mathcal{E},\mathsf{P})$ where $\Omega$ is a  finite   set of possible states of the world; 
$\mathcal{E}$ is its power set (i.e., the set of events); and $\mathsf{P}$ is a probability measure over $\Omega$.
We will consider two events $E_{A},E_{B}\in\mathcal{E}$ of interest to Bob and Alice, respectively (the choice of subscripts will become clear later). We will assume that they are perfectly correlated: 
$\mathsf{P}(E_{A}\backslash E_{B})=\mathsf{P}(E_{B}\backslash E_{A})=0.$

Fix partitions $\mathcal{P}_{A},\mathcal{P}_{B}$ of $\Omega$ for
Alice and Bob, respectively. For convenience, assume that all members of the join (coarsest
common refinement) of $\mathcal{P}_{A}$ and $\mathcal{P}_{B}$ are
non-null. For a state $\omega\in\Omega,$ $\mathcal{P}_{A}(\omega)$  ($\mathcal{P}_{B}(\omega)$) is the partition element of Alice's (Bob's) that contains $\omega.$ For each $n\in\mathbb{N},$ fix numbers $q_{A},q_{B}\in[0,1]$
and consider the following sets: 
\begin{align} 
\label{eq:A0B0} A_{0} & =\left\{ \omega\in\Omega:\mathsf{P}(E_{B}|\mathcal{P}_{A}(\omega))=q_{A}\right\}, \\
B_{0} & =\left\{ \omega\in\Omega:\mathsf{P}(E_{A}|\mathcal{P}_{B}(\omega))=q_{B}\right\}, \\
A_{n+1} & =\left\{ \omega\in A_{n}:\mathsf{P}(B_{n}|\mathcal{P}_{A}(\omega))=1\right\}, \\
B_{n+1} & =\left\{ \omega\in B_{n}:\mathsf{P}(A_{n}|\mathcal{P}_{B}(\omega))=1\right\} .
\label{eq:A0B0_end}
\end{align}
Here, the set $A_0$ is the set of states $\omega$ such that Alice assigns probability $q_A$ to event $E_B$; the set $B_1$ is the set of states $\omega$ such that Bob assigns probability $q_B$ to event $E_A$ and probability 1 to the states in $A_0$---i.e., states where Bob assigns probability $q_B$ to $E_A$ and is certain that Alice assigns probability $q_A$ to $E_B$; and so on, and similarly for $B_0$, $A_1$, etc.

If a state of the world $\omega^*$ is in sets $A_n$ and $B_n$, for all $n\in\mathbb{N}$, then, at $\omega^*$, Alice assigns probability $q_A$ to $E_B$, Bob is certain that Alice assigns probability $q_A$ to $E_B$, and Alice is certain that Bob is certain that Alice assigns probability $q_A$ to $E_B$, and Alice is certain that Bob is certain that... and so on indefinitely, and also vice versa about Bob assigning probability $q_B$ to $E_A$. In sum, there is common certainty at $\omega^*$ that Alice assigns probability $q_{A}$ to $E_{B}$ and that Bob assigns probability $q_{B}$ to $E_{A}$.

More formally, it is common certainty at a state $\omega^{*}\in\Omega$ that Alice
assigns probability $q_{A}$ to $E_{B}$ and that Bob assigns probability
$q_{B}$ to $E_{A}$ if
\begin{equation} \label{eq:cc-local}
\omega^{*}\in A_{n}\cap B_{n} \quad \forall n\in \mathbb{N}.
\end{equation}
If equation \eqref{eq:cc-local} does not hold for all $n\in \mathbb{N},$, but, instead, only for $n\leq N$ for a certain $N\in\mathbb{N},$ then we talk about $N$th-order mutual certainty.

We now state the classical agreement theorem that will be the basis of our work.  All proofs are given in the Supplementary Material.
\begin{thm} \label{thm:Aumann}
Fix a probability space $(\Omega,\mathcal{E},\mathsf{P})$, where $E_{A}$ and $E_{B}$ are perfectly
correlated events. If it is common certainty at a state $\omega^{*}\in\Omega$
that Alice assigns probability $q_{A}$ to $E_{B}$ and Bob assigns
probability $q_{B}$ to $E_{A},$ then $q_{A}=q_{B}.$
\end{thm}

\subsection{Defining agreement for observers of no-signaling systems}
We now map the classical agreement theorem into the no-signaling framework, in order to explore its applicability beyond the classical realm.

We consider no-signaling distributions, or boxes \cite{popescu_quantum_1994}, of the form
\begin{equation}
\left\{ p(ab|xy)\right\} _{a\in\mathcal{A},b\in\mathcal{B},x\in\mathcal{X},y\in\mathcal{Y}}\,,
\end{equation}
where $x,y$ and $a,b$ are Alice's and Bob's input and output, respectively, and $\mathcal{A},\mathcal{B},\mathcal{X},\mathcal{Y}$ are index
sets, not necessarily of the same size, and which satisfy
\begin{align}
    \sum_a p(ab|xy) &= \sum_a p(ab|x^{\prime} y),\\
    \sum_b p(ab|xy) &= \sum_b p(ab|xy^{\prime}),
\end{align}
 for all $x, x^{\prime}, y, y^{\prime}$. 

A no-signaling box is local if   there exist probability distributions \sloppy $\{p_{\lambda}:\lambda\in\Lambda\}$, $\{p_A(a|x\lambda):(a,x,\lambda)\in\mathcal{A}\times\mathcal{X}\times\Lambda\}$, $\,\{p_B(b|y\lambda):(b,y,\lambda)\in\mathcal{B}\times\mathcal{Y}\times\Lambda\}$, such that
\begin{equation} \label{eq:plocal}
    p(ab|xy)=\sum_{\lambda\in\Lambda} p_{\lambda} \, p_A(a|x\lambda) \, p_B(b|y\lambda),
\end{equation}
for each \sloppy $a,b,x,y,$ where $\Lambda$ is an index set.  
It is quantum if, for each $x,y$, there exist POVMs $\{E_x^a\}_{a\in\mathcal{A}}$, $\{F_y^b\}_{b\in\mathcal{B}}$ and a quantum state $\rho$ such that
\begin{equation}
    p(ab|xy)=\textnormal{tr}\left( E_x^a\otimes F_y^b \rho \right),
\end{equation}
for each $a,b$.
The set of local boxes is strictly included in the set of quantum boxes, which is, in turn, strictly included in the set of no-signaling boxes. No-signaling boxes that are not quantum are termed postquantum.

We now show that we can associate a no-signaling box with any ontological model, and vice versa. Let $\mathcal{A},\mathcal{B},\mathcal{X},\mathcal{Y}$ be index sets. Let $(\Omega,\mathcal{F},\mathsf{P})$ be a probability space, and, for each $x\in\mathcal{X}$, let $\{\mathsf{A}_{x}^{a}:a\in\mathcal{A}\}$ be a partition of the states $\omega\in \Omega$ where $a\in\mathcal{A}$ denotes the partition elements. Similarly, for each $y\in\mathcal{Y}$, let $\{\mathsf{B}_y^b:b\in\mathcal{B}\}$ be another partition of the states $\omega\in \Omega$, where $b\in\mathcal{B}$ denotes the partition elements. According to that, we can understand labels $x\in \mathcal{X}$, $y\in \mathcal{Y}$ as inputs---this information fixes what partition Alice and Bob look at---and $a \in \mathcal{A}$, $b\in \mathcal{B}$ as outputs---this is the information that the observers gain by observing their corresponding partitions. This terminology will shortly become very natural. 

With all the above, \sloppy $\big\lbrace(\Omega,\mathcal{F},\mathsf{P})$,
$\lbrace\mathsf{A}_{x}^{a},\mathsf{B}_{y}^{b}\rbrace_{a,b,x,y}\big\rbrace$ is an ontological model that we now want to  associate to a no-signaling box that reproduces its statistics.  In this ontological model, given inputs $x\in \mathcal{X}$, $y\in \mathcal{Y}$, the probability of obtaining outputs $a \in \mathcal{A}$, $b\in \mathcal{B}$ is given by $\mathsf{P}(\mathsf{A}_x^a \cap \mathsf{B}_y^b)$. This simple observation leads us to construct the no-signaling box
\begin{equation}
p(a,b|x,y):=\mathsf{P}\big(\mathsf{A}_{x}^{a}\cap\mathsf{B}_{y}^{b}\big),\qquad\forall(a,b,x,y)\in\mathcal{A}\times\mathcal{B}\times\mathcal{X}\times\mathcal{Y}.\label{OntToBox}
\end{equation}
It can be verified that the probabilities $p$ are non-negative, normalized, and no-signaling.

The converse process of finding  an ontological model starting from a no-signaling box can be also performed, as we show in Supplementary Note 2. Remarkably, this can be accomplished even in the case in which the no-signaling box is non-local, obtaining an ontological model with a quasi-probability measure  (i.e., one which allows for negative values, which still sum to 1) instead of standard positive probabilities \cite{abramsky_sheaf-theoretic_2011}. (The appearance of quasi-probabilities here should not surprise the reader. In fact, one cannot hope to obtain  ontological models with only non-negative probabilities for post-classical no-signaling boxes, since this would provide local hidden-variable models that contradict, for instance, Bell's theorem. In any case, the use of this mathematical tool has been well rooted in the study of quantum mechanics since its origins---see \cite{ferrie_quasi-probability_2011} for a nice review of this subject.) This makes it possible to translate  results from one framework to the other, something that might be of interest in order to establish further connections between epistemics and quantum theory. However, from now on, we focus on no-signaling boxes and leave this digression aside in the rest of the main text.   

With the  association   between ontological models and no-signaling boxes in mind, we next define common certainty of disagreement for no-signaling boxes. The idea is to reinterpret the definitions in Section \ref{sec:classical-agreement} in this latter setting. 

We first propose a meaning for the events of interest (previously identified as $E_A$, $E_B$) in the present setting. Now, these events correspond to some set of outcomes, given that the no-signaling box was queried with some particular inputs. For the sake of concreteness, we fix these inputs to be $x=1,\,  y=1$ and the outcomes of interest to be $a=1$, $b=1$. This motivates us to consider the events  $F_A=\lbrace(1,b,1,y):b\in \mathcal{B}, y\in \mathcal{Y}\rbrace$   (on Alice's side) and  $F_B = \lbrace(a,1,x,1):a\in \mathcal{A}, x\in \mathcal{X}\rbrace$   (on Bob's side). Then, we say that $F_A$ and $F_B$ are perfectly correlated when
\begin{equation} \label{eq:perfcorr}
p(a,b|x=1,y=1)=0 \textnormal{ for all } a\neq b.
\end{equation}

Given this, we assume that the observers actually conduct their measurements according to some partitions. Again, for concreteness, let us assume that those partitions are the ones associated with inputs $x=0$, $y=0$. These inputs take on the role of partitions $\mathcal{P}_A$, $\mathcal{P}_B$ in the ontological model picture. The outputs obtained from these measurements are the no-signaling box analogue to the events $\mathcal{P}_A(\omega)$, $\mathcal{P}_B(\omega)$. In order to make the following expressions more concrete, we assume, when $x=0$, $y=0$ are input, that the outputs obtained are $a=0$ and $b=0$, respectively. 

Therefore, given the perfectly correlated events $F_A,\, F_B$ and numbers $q_A,\, q_B \in [0,1]$, we define the sets
\begin{align}
\alpha_0 &=\left\{ a\in \mathcal{A} :p(b=1|a,x=0,y=1)=q_{A}\right\} \,,\\
\beta_0 &= \left\{ b\in \mathcal{B} :p(a=1|b,x=1,y=0)=q_{B}\right\} \, ,
\end{align}
and, for all $n\geq 0,$
\begin{align}
\alpha_{n+1} &= \left\{ a \in \alpha_n :p(B_n|a,x=0,y=0)=1\right\} \, ,\\
\beta_{n+1}  &= \left\{ b \in \beta_n :p(A_n |b,x=0,y=0)=1\right\} \, ,
\end{align}
where 
\begin{align}
A_n &= \alpha_n\times \mathcal{B} \times \mathcal{X} \times \mathcal{Y}\,,\\
B_n &= \mathcal{A} \times \beta_n \times \mathcal{X} \times \mathcal{Y}\,.
\end{align}
By analogy with the sets in equations \eqref{eq:A0B0}-\eqref{eq:A0B0_end}, the set $\alpha_0$ consists of Alice's outcomes such that she assigns probability $q_A$ to $F_B$  upon seeing that outcome, having input $x=0$, if she assumes that Bob inputs $y=1$. The set $\beta_1$ is the set of Bob's outcomes such that he is certain that Alice assigns probability $q_A$ to $F_B$  upon seeing that outcome, having input $y=0$, if he assumes that Alice inputs $x=1$, and so on, and similarly for $\beta_0$, $\alpha_1$, etc.

With this in mind, we can build a chain of mutual certainties, in a manner similar to the classical agreement theorem. If the event $(a=0, b=0, x=0, y=0)$ is in the sets $A_n$ and $B_n$, for all $n\in\mathbb{N}$, then, if Alice and Bob both input 0 and get output 0, we have: Alice assigns probability $q_A$ to $F_B$ (assuming that Bob input $y=1$), Bob is certain that Alice assigns probability $q_A$ to $F_B$, Alice is certain that Bob is certain that... and so on indefinitely, and vice versa. That is, there is common certainty at $(a=0, b=0, x=0, y=0)$ that Alice assigns probability $q_A$ to $F_B$ and that Bob assigns probability $q_B$ to $F_A$.

More formally, there is common certainty about the event that Alice assigns probability $q_A$ to $F_B$ and that Bob assigns probability $q_B$ to $F_A$  if 
\begin{equation}\label{def:ccd_in_NS}
(a=0,b=0,x=0,y=0) \in A_n \cap B_n \quad \forall n\in \mathbb{N}. 
\end{equation}
There is common certainty of disagreement if, in addition, $q_A \neq q_B$.

Notice the relationship between this definition and the previous one: the $\omega^*$ in equation \eqref{eq:cc-local}, at which the disagreement occurred, fixed the partition elements that Alice and Bob observed. Here, disagreement occurs at the inputs and outputs $(a=0, b=0, x=0, y=0)$ that the observers obtain.

 In complete generality, we can also consider disagreement at arbitrary inputs and outputs $(a,b,x,y)$. For that, one just has to consider the appropriate changes in the preceding paragraphs. The only case that might seem different is that in which Alice and/or Bob's input is the same as that corresponding to the event of interest; that is, $x=1$ and/or $y=1$. This is allowed but uninteresting: It is easy to see that the fact that $F_A$ and $F_B$ are perfectly correlated precludes the possibility of common certainty of disagreement in this case. Therefore, for the sake of concreteness, we fix  $x$, $y$ both different from $1$ and, in particular, equal to $0$.   

The following is a rephrasing of the classical agreement theorem: 
\begin{thm}\label{cor:Aumann}
	Suppose that Alice and Bob share a local no-signaling box with underlying probability
	distribution $p$. Let $q_A,q_B \in [0,1]$, and let
	\begin{align}
	    p(b=1|a=0,x=0,y=1) &= q_A,\\
	    p(a=1|b=0,x=1,y=0) &= q_B.
	\end{align}
	If $q_{A}$ and $q_{B}$ are common certainty between the observers,	then $q_{A}=q_{B}.$
\end{thm}

In Supplementary Note 3, we give a standalone proof of this result.
 Moreover, using the above correspondence between ontological models and classical no-signaling boxes, one can prove that the notions of common certainty of disagreement in Theorems \ref{thm:Aumann} and \ref{cor:Aumann} are equivalent.
  We now ask whether this theorem holds in quantum and no-signaling settings.

\subsection{Observers of quantum systems cannot agree to disagree}

Given the mapping exhibited above, as well as the restatement of the agreement theorem for local boxes, it is now natural to ask whether the theorem holds when dropping the locality constraint. 

We address this question by exploring it in the broader no-signaling setting. First, we establish that, in general, observers of no-signaling systems can agree to disagree about perfectly correlated events, and we give explicit examples of disagreeing no-signaling distributions. In the particular case of two inputs and two outputs, we characterize the distributions that give rise to common certainty of disagreement. One might think that the fact that observers of no-signaling systems can agree to disagree is a direct consequence of the multitude of uncertainty relations in quantum mechanics, all of which put a limit on the precision with which the values of incompatible
observables can be measured and which have even been linked to epistemic inconsistencies in quantum mechanics \cite{frauchiger_quantum_2018}. Somewhat surprisingly, our next finding shows that this is not the case. We show that disagreeing no-signaling distributions of two inputs and two outputs cannot be quantum---i.e., the agreement theorem holds for observers of quantum systems in this setting. Then, we go beyond this restriction and show that any disagreeing no-signaling distribution with more than two inputs or outputs induces a disagreeing distribution with two inputs and outputs. Since the agreement theorem holds for observers of quantum systems sharing distributions of two inputs and outputs, it does so for more general distributions too. Thus, even if quantum mechanics features uncertainty relations, this does not apply to observers' estimates of perfectly correlated events.

We first present the following theorem in which the no-signaling box has two inputs and two outputs, but we will show in Theorem~\ref{thm:generalization} that the result is fully general. In place of ``common certainty of disagreement about the event that Alice assigns probability $q_A$ to  $F_B = \lbrace(a,1,x,1):a\in \mathcal{A}, x\in \mathcal{X}\rbrace$   and Bob assigns probability $q_B$ to  $F_A=\lbrace(1,b,1,y):b\in \mathcal{B}, y\in \mathcal{Y}\rbrace$  , at event $(0, 0, 0, 0)$,'' we simply say ``common certainty of disagreement.''

\begin{thm} \label{thm:2in2outcc}
A two-input two-output no-signaling box gives rise to common certainty of disagreement if and only if it takes the form of Table \ref{tab:nsccd}.
\end{thm}
\begin{table}[ht]
\centering
\begin{tabular}{|c|c|c|c|c|}
\hline 
\label{isabox}$xy\backslash ab$  & 00  & 01  & 10  & 11\tabularnewline
\hline 
\hline 
00  & $r$  & 0  & 0  & $1-r$ \tabularnewline
\hline 
01  & $r-s$  & $s$  & $-r+t+s$  & $1-t-s$\tabularnewline
\hline 
10  & $t-u$  & $u$  & $r-t+u$  & $1-r-u$\tabularnewline
\hline 
11  & $t$  & 0  & 0  & $1-t$\tabularnewline
\hline 
\end{tabular}
\caption{Parametrization of two-input two-output no-signaling boxes with common certainty of disagreement. Here, $r,s,t,u  \in [0,1]$ are such that all the entries of the box are non-negative, $r>0$, and $s-u\neq r-t.$}
\label{tab:nsccd}
\end{table}
While some no-signaling distributions can exhibit common certainty of disagreement, we find that probability distributions arising in quantum mechanics do satisfy the agreement theorem. This is surprising: It is well-known that a given measurement of a quantum system (say, that corresponding to the input $x,y=0$) need not offer any information about the outcome of an incompatible measurement on the same system (say, $x,y=1$). However, some consistency remains: common certainty of disagreement is impossible, even for incompatible measurements.

\begin{thm} \label{thm:ccd-not-quantum}
No two-input two-output quantum box can give rise to common certainty of disagreement.
\end{thm}
The proof follows by deriving a contradiction from Tsirelson's theorem~\cite{cirelson_quantum_1980}.

We have seen that no two-input two-output quantum box can give rise to common certainty of disagreement. We now lift the restriction on the number of inputs and outputs and show that no quantum box can give rise to common certainty of disagreement.

 First, as we now show, the proof for two inputs and outputs does not require common certainty, but only first-order mutual certainty. By observing the definitions of the sets $\alpha_n, \beta_n$, one can see that $\alpha_n=\alpha_1$ and $\beta_n=\beta_1$ for all $n\geq 1$. This means that first-order mutual certainty implies common certainty, and, therefore, first-order certainty suffices to characterize the no-signaling box that displays common certainty of disagreement.

As the number of outputs grows, first-order mutual certainty is no longer sufficient. However, since the number of outputs is always finite, there exists an $N\in\mathbb{N}$ such that $\alpha_n=\alpha_N$ and $\beta_n=\beta_N$ for all $n\geq N$. Since $\alpha_{n+1}\subseteq \alpha_n \, \forall n$, and similarly for $\beta$, the sets $\alpha_N, \beta_N$ are the smallest sets of outputs for which the disagreement occurs. Because of this, any $(a,b,x,y)$ that belongs to $A_N \cap B_N$ will also belong to $A_n \cap B_n$ for all $n$; that is, $N$th-order mutual certainty implies common certainty. So, for any finite no-signaling box, one needs only $N$th-order mutual certainty to characterize it. As the number of outputs grows unboundedly, one needs common certainty to hold \cite{geanakoplos_we_1982}. These observations will be relevant to extending Theorem \ref{thm:ccd-not-quantum} beyond two inputs and outputs.

\begin{thm} \label{thm:generalization}
No quantum box can give rise to common certainty of disagreement.
\end{thm}

To prove the theorem, we show that any no-signaling box with common certainty of disagreement induces a two-input two-output no-signaling box with the same property. Thus, if there existed a quantum system that could generate the bigger box, it could also generate the smaller box. Then, Theorem \ref{thm:2in2outcc} implies that no quantum box can give rise to common certainty of disagreement.

\subsection{Observers of quantum systems cannot disagree singularly}

Next, we ask if observers of no-signaling quantum systems can disagree in other ways. We define a new notion of disagreement, which we call singular disagreement, by removing the requirement of common certainty and, instead, imposing $q_{A}=1,$ $q_{B}=0$. We ask whether this new notion holds for observers of classical, quantum, and no-signaling systems. We find the same pattern as before: Singular disagreement does not hold for observers of classical or quantum systems, but can occur in no-signaling settings, where we characterize the distributions that feature it.

Suppose that Alice and Bob share a no-signaling box. As before, Alice assigns probability $q_A$ to the event $F_B$, and Bob assigns probability $q_B$ to the event $F_A$. Suppose that this happens at event $(a=0,b=0,x=0,y=0)$. If $q_A=1$ and $q_B=0$, these probabilities differ maximally: Alice is certain that $F_B$ happens, while Bob is certain that $F_A$ does not happen. If, in addition, $F_A$ and $F_B$ are perfectly correlated, then there is singular disagreement at $(0,0,0,0)$.
This time, there is no need for chained certainties---we just require that Alice's and Bob's assignments differ maximally.

More formally, there is singular disagreement about the probabilities assigned by Alice and Bob to perfectly correlated events  $F_B = \lbrace(a,1,x,1):a\in \mathcal{A}, x\in \mathcal{X}\rbrace$   and  $F_A=\lbrace(1,b,1,y):b\in \mathcal{B}, y\in \mathcal{Y}\rbrace$,   respectively, at event $(0, 0, 0, 0)$ if it holds that 
\begin{equation}
q_{A}=1,\;q_{B}=0.\label{eq:singdiscondt}
\end{equation}

Similarly to the previous section, we refer to the above definition simply as ``singular disagreement.''

We restrict ourselves first to boxes of two inputs and outputs and show that local boxes cannot exhibit singular disagreement. Then, we characterize the no-signaling boxes that do satisfy singular disagreement and show they cannot be quantum. Finally, we generalize to boxes of any number of inputs and outputs.

\begin{thm} \label{thm:sd-local}
\label{prop:loc-sd}There is no local two-input two-output box that gives rise to singular disagreement. 
\end{thm}

It turns out that singular disagreement induces a Hardy paradox \cite{hardy_quantum_1992} in the system; therefore it cannot be local.

We now lift the local restriction and characterize the no-signaling boxes in which singular disagreement occurs.
\begin{thm} \label{thm:singdis}
A two-input two-output no-signaling box gives rise to singular disagreement if and only if it takes the form of Table \ref{tab:nssd}.
\end{thm}
\begin{table}[ht!]
\centering
\begin{tabular}{|c|c|c|c|c|}
\hline 
$xy\backslash ab$  & 00  & 01  & 10  & 11\tabularnewline
\hline 
\hline 
00  & $s$  & $t$  & $1-s-u-t$  & $u$\tabularnewline
\hline 
01  & 0  & $s+t$  & $r$  & $1-s-t-r$\tabularnewline
\hline 
10  & $1-u-t$  & $u+t+r-1$  & 0  & $1-r$\tabularnewline
\hline 
11  & $r$  & 0  & 0  & $1-r$\tabularnewline
\hline 
\end{tabular}
\caption{Parametrization of two-input two-output no-signaling boxes with singular disagreement. Here, $r,\,s,\, t,\, u,\, \in [0,1]$ are such that all the entries of the box are non-negative, $s>0$, and $s + t \neq 0$ and $u+t\neq 1$.}
\label{tab:nssd}
\end{table}

The definition of singular disagreement gives rise to the zeros in the second and third rows of Table \ref{tab:nssd}, while the zeros in the bottom row come from the perfectly correlated outputs on inputs $x=y=1$.

However, singular disagreement cannot arise in quantum systems. This is another way in which quantum mechanics provides some consistency between (possibly incompatible) measurements, just as in the case of common certainty of disagreement. 
\begin{thm}
No  two-input  two-output quantum box  can  give  rise  to  singular disagreement.
\end{thm}
By Ref. \cite{rai_geometry_2019}, the boxes of Theorem \ref{thm:singdis} are either local or postquantum. But Theorem \ref{thm:sd-local} implies they cannot be local.

Finally, the above results can be generalized to any finite box:
\begin{thm}\label{thm:sd-not-quantum}
No quantum box  can  give  rise  to  singular disagreement.
\end{thm}
The proof is very similar to that of Theorem \ref{thm:generalization}.

Additionally, the PR box \cite{popescu_quantum_1994} satisfies both Theorems \ref{thm:2in2outcc} and \ref{thm:singdis}, and this makes it an example of both kinds of disagreement.

\section{Discussion}

We have defined two notions of disagreement inspired by notions from epistemics. Both notions of disagreement imply immediate tests for new theories---namely, the tables in Theorem \ref{thm:2in2outcc} and Theorem \ref{thm:singdis}. These tests are very general in the sense that they are based only on the capability (or not) of a theory to realize undesirable correlations between non-communicating parties. Also, both principles have their roots in epistemics, with common certainty of disagreement closer to Aumann's original idea and singular disagreement permitting a simpler description.

These two notions of disagreement are compatible in that it is possible to find examples displaying both kinds of disagreement at once. Strikingly, a prime such example is the Popescu-Rohrlich box \cite{popescu_quantum_1994}, establishing that it is an extremal resource both in the sense of being an extreme point of the polytope of no-signaling distributions and in the sense of inducing the strongest possible disagreement between two parties.

On a speculative note, we suggest that it would be also very interesting to explore the application of the concepts introduced in this paper to practical tasks in which consensus between parties plays a role, such as the coordination of the action of distributed agents or the   
verification of distributed computations. See \cite{book_ReasoningAboutKnowledge} for some specific connections along these lines in the classical case.

Further work could be dedicated to constructing physical paradoxes arising from the possibility of agreeing to disagree. The examples in the literature (see Refs. \cite{geanakoplos_sebenius_1983,milgrom_stokey_1982,aumann_brandenburger_1995}), while very practical, might, to some communities, be considered less appealing than deeper physical consequences. These could be explored by exploiting the newly built bridges between quantum mechanics and epistemics.

Our results suggest that agreement can be used to design experiments to test the behaviour of Nature. In experimental settings, noise is unavoidable. Adding white noise to the boxes in Table \ref{tab:nsccd} and Table \ref{tab:nssd} (both of which lie in quantum voids) would mean that the zeros in the boxes now become small but nonzero parameters.  Robustness of quantum voids to this type of noise can be deduced from the closure of the set of quantum 2-input 2-output correlations \cite{goh_2018}. This already covers an approximate version of singular disagreement.   Another future direction to explore would be defining notions of approximate common certainty of disagreement.

We contend that agreement between observers could be a convenient principle for testing the consistency of new postquantum theories. Our results yield a clear parameterization 
of the set of the probability distributions that allow observer disagreement. This set is easy to work with,
thanks to its restriction to two observations and two outcomes per observer.
If a new theory can be used to generate such a distribution, this might raise a red flag, since this 
theory violates a reasonable, intuitive and, importantly, testable property that quantum mechanics satisfies.

Therefore, another future direction would be to study disagreement in theories that generalize quantum theory. For instance, one could consider almost quantum correlations \cite{navascues_almost-quantum_2015}, which is a set of correlations strictly larger than those achievable by measuring quantum states but that were designed to satisfy all physical principles previously proposed in the literature.
Almost quantum correlations are well characterized in terms of no-signaling boxes, so a natural question is whether they permit common certainty of disagreement or singular disagreement. For common certainty of disagreement, a straightforward modification to the proof of Theorem \ref{thm:ccd-not-quantum} shows that this phenomenon cannot arise under almost quantum correlations.  For singular disagreement, the simple characterization of almost quantum correlations in terms of a semidefinite program \cite{navascues_2008} makes it possible to search numerically for almost quantum boxes displaying singular disagreement.  Using this approach, we have found numerical evidence that singular disagreement, too, cannot arise under almost quantum correlations. Hence, our principles would appear to make it possible to identify additional features that almost quantum correlations share with quantum theory, indicating a new sense in which these generalized correlations are still reasonable.

Furthermore, one of the aims of studying physical principles is to find out whether quantum theory can be deduced from a set of such principles. Our findings open new directions for the exploration of such a programme. 
However, there are two main challenges in doing so.
 First, in order to assess this, it would be desirable to compare this principle to existing ones in the literature \cite{popescu_quantum_1994,clifton_characterizing_2003,pawlowski_information_2009,sun_no_2019,yan_quantum_2013}. However, our principle is markedly different in nature---this novelty being precisely the main obstacle. Second, Ref. \cite{gallego_quantum_2011} shows that bipartite principles are not enough to capture the set of quantum correlations. This highlights the interest of finding a multipartite extension for our principle. 
In the classical case, multipartite extensions have already been considered (see, e.g. \cite{parikh_communication_1990}). Studying similar extensions in the quantum and no-signaling settings is an interesting avenue for future work.

Finally, we compare our results with others in the physics literature that call on Aumann's theorem. 
Refs. \cite{khrennikov_quantum_2015,khrennikov_possibility_2014} examine Aumann's theorem when observers
are assumed to use Born's rule to update probabilities. The
authors conclude that Aumann's theorem does not hold for this type
of observer. Instead, our setting assumes that the observers are macroscopic
and merely share a quantum state or a no-signaling box. Our setting
is appealing in quantum information for its applications to communication
complexity, cryptography, teleportation, and many other scenarios. Ref. \cite{abramsky_non-locality_2019} introduces a different
notion of disagreement in a no-signaling context. Here, disagreement concerns pieces of information about some variables,
and agreement refers to consistency in the information provided about
the variables. Hence, it is unrelated to the epistemic notion of disagreement
in Aumann's theorem or in the present work.  Disagreement is also a feature of Ref. \cite{frauchiger_quantum_2018}. However, there, one of the observers ignores the fact that a certain measurement has taken place, while another observer takes this into account, and disagreement arises over the outcomes of a different measurement. In our work, the events that take place are the same according to all observers.

\begin{comment}However, even in the classical case, several extensions of the two-observer setup are possible, with different epistemic structures. Some give rise to agreement to disagree, others do not (see, e.g., \cite{parikh_communication_1990}). It would be interesting to explore these epistemic structures in the quantum and no-signaling settings too.\end{comment}

\paragraph{Data availability.} Data sharing not applicable to this article since no datasets were generated or analysed during the current study.

%\bibliographystyle{plain}
%\bibliography{ADNCW}

\paragraph{Acknowledgements.} 
We thank Eduardo Zambrano and Valerio Scarani for helpful communications. We also thank our audiences at Perimeter Institute, Universidad Complutense de Madrid, QuSoft, and IQOQI in Vienna for their comments and illuminating discussions. We are grateful to Elie Wolfe for bringing Ref. \cite{rai_geometry_2019} to our attention, Alex Pozas-Kerstjens for help with the numerical search, Miguel Navascu\'es for spotting a mistake in a previous version of the paper, Ronald de Wolf, Serge Fehr, Rena Henderson,  Laura Man\v cinska, and Peter van Emde Boas.
We acknowledge financial support from NYU Stern School of Business, NYU Shanghai,
and J.P. Valles (A.B.); MTM2017-88385-P funded by MCIN/AEI/ 10.13039/501100011033 and by ``ERDF A way of making Europe'', QUITEMAD-CM P2018/TCS4342 (Comunidad de Madrid), SEV-2015-0554-16-3 funded by MCIN/AEI/ 10.13039/501100011033, and PID2020-113523GB-I00 funded by MCIN/AEI/ 10.13039/501100011033 and by ``ERDF A way of making Europe'' (P.C-T.);
ERC Consolidator Grant GAPS (No. 648913) and MTM2017-83262-C2-1-P funded by MCIN/AEI/ 10.13039/501100011033 and by ``ERDF A way of making Europe'' (A.M.K.); Deutsche Bundesbank (P.L.M.);
MTM2014- 54240-P funded by MCIN/AEI/ 10.13039/501100011033 and by ``ERDF A way of making Europe'', QUITEMAD-CM P2018/TCS4342 (Comunidad de Madrid), ICMAT Severo Ochoa project SEV-2015-0554 funded by MCIN/AEI/ 10.13039/501100011033,
and PID2020-113523GB-I00 funded by MCIN/AEI/ 10.13039/501100011033 and by ``ERDF A way of making Europe'', "Acci\'on financiada por la Comunidad de Madrid en el marco del Convenio Plurianual con la Universidad Polit\'ecnica de Madrid en la l\'inea de actuaci\'on Programa de Excelencia para el Profesorado Universitario" (G.S.). Contreras-Tejada is
grateful for the hospitality of Perimeter Institute, where part of this work was carried out. Research at Perimeter
Institute is supported, in part, by the Government of Canada through the Department of Innovation, Science and
Economic Development Canada and by the Province of Ontario through the Ministry of Economic
Development, Job Creation and Trade.

\paragraph{Author contributions.}
P.C.-T. and G.S. contributed equally to this work. P.C.-T., G.S., A.M.K., A.B. and P.L.M. contributed to designing the ideas, performing the calculations, analyzing the results, and writing the manuscript.

\paragraph{Competing interests.} The authors declare no competing interests.

%\section*{Tables}

%\end{document}

\nocite{abramsky_operational_2014}
\bibliographystyle{plain}
\bibliography{abo-bibliography.bib}

\cleardoublepage1

\setcounter{thm}{0}
\setcounter{equation}{0}
\setcounter{table}{0}
\begin{center}
\large{\textbf{Supplementary material}}
\end{center}
\appendix

We prove the results in the main text. We also provide the statements and some definitions in the interest of readability.

\section*{Supplementary Note 1 - Proof of Theorem 1}

\begin{thm}\label{thm:Aumann-app}
Fix a probability space $(\Omega,\mathcal{E},\mathsf{P})$ where $E_{A}$ and $E_{B}$ are perfectly
correlated events. If it is common certainty at a state $\omega^{*}\in\Omega$
that Alice assigns probability $q_{A}$ to $E_{B}$ and Bob assigns
probability $q_{B}$ to $E_{A},$ then $q_{A}=q_{B}.$
\end{thm}

\begin{proof}
Since $\Omega$ is finite, there is a finite $N\in\mathbb{N}$ such
that, for all $n\geq N,$ $A_{n+1}=A_{n}$ and $B_{n+1}=B_{n}.$ From
the definition of $A_{N+1},$ we have that
\begin{equation}
\mathsf{P}(B_{N}|\mathcal{P}_{A}(\omega))=1\ \forall\omega\in A_{N}.\label{eq:A_N}
\end{equation}
Now, $A_{N}$ is a union of partition elements of $\mathcal{P}_{A},$
i.e., $A_{N}=\bigcup_{i\in I}\pi_{i}$ where each $\pi_{i}\in\mathcal{P}_{A}$
and $I$ is a finite index set. From Equation (\ref{eq:A_N}), we
have
\begin{equation}
\mathsf{P}(B_{N}|\pi_{i})=1\ \forall i\in I.
\end{equation}
Since $\mathsf{P}(B_{N}|A_{N})$ is a convex combination of $\mathsf{P}(B_{N}|\pi_{i})$
for $i\in I,$ we must have
\begin{equation}
\mathsf{P}(B_{N}|A_{N})=1.\label{eq:probBNAN}
\end{equation}

Now, since $A_{N}\subseteq A_{0},$ then $\mathsf{P}(E_{B}|\pi_{i})=q_{A}$
for all $i\in I$ too. Using a convex combination argument once more,
this entails that
\begin{equation}
\mathsf{P}(E_{B}|A_{N})=q_{A}.\label{eq:probEBAN}
\end{equation}
Equations (\ref{eq:probBNAN}) and (\ref{eq:probEBAN}) together imply
that
\begin{equation}
\mathsf{P}(E_{B}|A_{N}\cap B_{N})=q_{A}.
\end{equation}

But events $E_{A}$ and $E_{B}$ are perfectly correlated, so that 
\begin{equation}
\mathsf{P}(E_{A}\cap E_{B}|A_{N}\cap B_{N})=q_{A},
\end{equation}
as well.

Running the parallel argument with $A$ and $B$ interchanged, we
obtain
\begin{equation}
\mathsf{P}(E_{A}\cap E_{B}|A_{N}\cap B_{N})=q_{B},
\end{equation}
which implies that $q_{A}=q_{B}.$
\end{proof}

\section*{Supplementary Note 2 - Relating no-signalling boxes and ontological models}

The following result already appeared elsewhere \cite{abramsky_sheaf-theoretic_2011}. For convenience, we write down the proof in terms of elementary linear algebra without relying on the sheaf-theoretic language of Ref. \cite{abramsky_sheaf-theoretic_2011}.

%%The following  result was already derived in \cite{abramsky_sheaf-theoretic_2011} from sheaf-theoretic concepts, however we provide a  proof that does not rely on sheaf theory and \color{black} is more suitable for the purposes of this work. 

\noindent\textit{Proposition A1.}{ Given any no-signaling box  $\{p(ab|xy):(a,b,x,y)\in\mathcal{A}\times\mathcal{B}\times\mathcal{X}\times\mathcal{Y}\},$ there is a (non-unique) corresponding ontological model whose probabilities assigned to the states of the world are not necessarily non-negative.
}

\begin{proof}
Let  $\{p(ab|xy):(a,b,x,y)\in\mathcal{A}\times\mathcal{B}\times\mathcal{X}\times\mathcal{Y}\}$  be a no-signaling box. We construct its associated ontological model $\big\lbrace(\Omega,\mathcal{F},\mathsf{P})$,
$\lbrace\mathsf{A}_{x}^{a},\mathsf{B}_{y}^{b}:(a,b,x,y)\in \mathcal{A}\times\mathcal{B}\times\mathcal{X}\times\mathcal{Y} \rbrace \big\rbrace\,.$ We provide the proof for $a,b,x,y\in\left\{ 0,1\right\} \,$ for ease of notation, but the generalization to more inputs and outputs is immediate.

To construct the ontological model, we postulate
the existence of a set of states 
\begin{equation}
\omega_{a_{0}a_{1}b_{0}b_{1}}
\end{equation}
with quasi-probabilities 
\begin{equation}
\mathsf{P}_{a_{0}a_{1}b_{0}b_{1}}\equiv\mathsf{P}(\omega_{a_{0}a_{1}b_{0}b_{1}}).
\end{equation}
Each state corresponds to an \textit{instruction set} \cite{abramsky_operational_2014}, i.e., the state where Alice outputs $a_{0}$ on input $x=0$ and $a_{1}$ on input $x=1\,,$ and Bob outputs $b_{0}$ on input $y=0$ and $b_{1}$ on input $y=1.$ Then, each  $\mathsf{P}_{a_{0}a_{1}b_{0}b_{1}}$ is the quasi-probability of the corresponding instruction set. Of course, if the given box is post-classical, not all of these quasi-probabilities will be non-negative. In fact, in principle it need not even be guaranteed that one can find a quasi-probability distribution over these states. But we will use the probability distribution of the inputs and outputs of the given no-signaling box to derive a linear system of equations over the quasi-probabilities, and show that it does have a solution.

There are 16 states
in total, as there are two possible outputs for each of the 4 inputs
($|\mathcal{A}|^{|\mathcal{X}|}\cdot|\mathcal{B}|^{|\mathcal{Y}|}$
in general). Then, each partition corresponds to a set of states
as follows: 
\begin{equation}
\begin{aligned}\mathsf{A}_{x}^{a} & =\left\{ \omega_{a_{0}a_{1}b_{0}b_{1}}:a_{x}=a\right\} \\
\mathsf{B}_{y}^{b} & =\left\{ \omega_{a_{0}a_{1}b_{0}b_{1}}:b_{y}=b\right\} \,.
\end{aligned}
\label{eq:defA,B}
\end{equation}

We associate the probabilities $p(ab|xy)$
of the no-signaling box to the probabilities $\mathsf{P}(\mathsf{A}_{x}^{a}\cap\mathsf{B}_{y}^{b})$
of each intersection of partitions, for each input pair
$(x,y)$ and output pair $(a,b)\,.$ This gives rise to a set of equations
for the probabilities $\mathsf{P}_{a_{0}a_{1}b_{0}b_{1}}.$
Indeed, the probability of each intersection is given by 
\begin{equation}
\mathsf{P}(\mathsf{A}_{x}^{a}\cap\mathsf{B}_{y}^{b})=\sum_{a_{\bar{x}}\,,b_{\bar{y}}}\mathsf{P}_{a_{x}a_{\bar{x}}b_{y}b_{\bar{y}}}\label{eq:prob-intersection}
\end{equation}
where we denote the output corresponding to the input that is not
$x$ as $a_{\bar{x}}\,,$ and similarly for $b_{\bar{y}}\,,$ and
so we have, for each $a,b,x,y\,,$ 
\begin{equation}
\sum_{a_{\bar{x}}\,,b_{\bar{y}}}\mathsf{P}_{a_{x}a_{\bar{x}}b_{y}b_{\bar{y}}}=p(ab|xy)\,.\label{eq:linsyst}
\end{equation}
Since there are 16 values of $p(ab|xy)$ in the 2-input 2-output no-signaling
box, we arrive at 16 equations ($|\mathcal{A}|\times|\mathcal{B}|\times|\mathcal{X}|\times|\mathcal{Y}|$
in general). Of course, there are some linear dependencies between
the equations, but we will show that the system still has a solution.

The system of equations can be expressed as 
\begin{equation}
MP=C
\end{equation}
where $M$ is the matrix of coefficients, $P$ is the vector of probabilities
$\mathsf{P}_{a_{0}a_{1}b_{0}b_{1}}$ and $C$ is the vector of independent
terms $p(ab|xy)\,.$ The system has a solution (which is not necessarily
unique) if and only if 
\begin{equation}
\textnormal{rank}(M)=\textnormal{rank}(M|C)\,.
\end{equation}
Since the rank of a matrix is the number of linearly independent rows,
it is trivially true that 
\begin{equation}
\textnormal{rank}(M)\leq\textnormal{rank}(M|C)\,,
\end{equation}
as including the independent terms can only remove some relations
of linear dependence, not add more. Equivalently, the number of relations
of linear dependence of $M|C$ is always smaller than or equal to
the number of relations of linear dependence of $M$. Therefore,
to show that their ranks are equal, it is sufficient to show that
every relation of linear dependence that we find in $M$ still holds
in $M|C$. That is, for every relation of linear dependence between
the probabilities $\mathsf{P}_{a_{0}a_{1}b_{0}b_{1}}$ that is contained
in $M\,,$ it is sufficient to show that the relation still holds
when the sums of probabilities are matched to the elements $p(ab|xy)$
of the no-signaling box in order to show that the system of equations has a
solution.

Observe that $M$ contains only zeros and ones, as the equations (\ref{eq:linsyst})
are just sums of probabilities. Moreover, each column of $M$ corresponds
to the probability of a state $\omega_{a_{0}a_{1}b_{0}b_{1}}\,,$
while each row corresponds to an equation with independent term $p(ab|xy)\,.$
Because the equations (\ref{eq:linsyst}) correspond to intersections
of partitions of the set of $\omega_{a_{0}a_{1}b_{0}b_{1}}\,,$ we
can observe that each row of $M$ has a 1 in the column corresponding
to the states $\omega_{a_{0}a_{1}b_{0}b_{1}}$ contained in the
corresponding partition, and a 0 elsewhere. Put another way, in order
to construct $M$ one must first partition the set of $\omega_{a_{0}a_{1}b_{0}b_{1}}$
in four different ways, corresponding to 
\begin{equation}
\begin{aligned}\left\{ \mathsf{A}_{0}^{a}: a\in \{0,1\}\right\}\,, & \left\{ \mathsf{A}_{1}^{a}:a\in \{0,1\}\right\}\,,\\
\left\{ \mathsf{B}_{0}^{b}:b\in \{0,1\}\right\}\,, & \left\{ \mathsf{B}_{1}^{b}:b\in \{0,1\}\right\}
\end{aligned}
\end{equation}
for Alice and Bob respectively. This gives partitions of the columns
of $M\,.$ Then, the 16 possible ways of intersecting partitions of
Alice's with partitions of Bob's give the 16 equations with independent
term $p(ab|xy)\,.$ But notice now that the partition structure imposes
a certain relation of linear dependence between the rows of $M.$
Indeed, for each $b,y,$ we have 
\begin{equation}
\left\{ \bigcup_{a}\left(\mathsf{A}_{0}^{a}\cap\mathsf{B}_{y}^{b}\right)\right\} =\left\{ \bigcup_{a}\left(\mathsf{A}_{1}^{a}\cap\mathsf{B}_{y}^{b}\right)\right\} \,,
\end{equation}
as 
\begin{equation}
\bigcup_{a}\left(\mathsf{A}_{0}^{a}\cap\mathsf{B}_{y}^{b}\right)=\left(\bigcup_{a}\mathsf{A}_{0}^{a}\right)\cap\mathsf{B}_{y}^{b}=\Omega\cap\mathsf{B}_{y}^{b}=\left(\bigcup_{a}\mathsf{A}_{1}^{a}\right)\cap\mathsf{B}_{y}^{b}=\bigcup_{a}\left(\mathsf{A}_{1}^{a}\cap\mathsf{B}_{y}^{b}\right)\,,
\end{equation}
and, similarly, for each $a,x$ we have 
\begin{equation}
\left\{ \bigcup_{b}\left(\mathsf{A}_{x}^{a}\cap\mathsf{B}_{0}^{b}\right)\right\} =\left\{ \bigcup_{b}\left(\mathsf{A}_{x}^{a}\cap\mathsf{B}_{1}^{b}\right)\right\} \,.
\end{equation}
Using the correspondence of these partitions with the partitions of
the columns of $M$ gives 8 relations of linear dependence between
its rows. Now, noticing that a union of columns of $M$ corresponds
to a sum of probabilities $\mathsf{P}_{a_{0}a_{1}b_{0}b_{1}}\,,$
we find that these relations correspond exactly to the no-signaling conditions,
as 
\begin{equation}
\mathsf{P}\left(\bigcup_{a}\left(\mathsf{A}_{x}^{a}\cap\mathsf{B}_{y}^{b}\right)\right)=\sum_{a}\mathsf{P}\left(\mathsf{A}_{x}^{a}\cap\mathsf{B}_{y}^{b}\right)
\end{equation}
so 
\begin{equation}
\sum_{a}\mathsf{P}\left(\mathsf{A}_{0}^{a}\cap\mathsf{B}_{y}^{b}\right)=\sum_{a}\mathsf{P}\left(\mathsf{A}_{1}^{a}\cap\mathsf{B}_{y}^{b}\right)
\end{equation}
and similarly for Bob. Of course, by definition of no-signaling box, these relations
hold for the independent terms $p(ab|xy)$ as well, since 
\begin{equation}
\begin{aligned}\sum_{a}\mathsf{P}\left(\mathsf{A}_{x}^{a}\cap\mathsf{B}_{y}^{b}\right) & =\sum_{a}p(ab|xy)\,,\\
\sum_{b}\mathsf{P}\left(\mathsf{A}_{x}^{a}\cap\mathsf{B}_{y}^{b}\right) & =\sum_{b}p(ab|xy)
\end{aligned}
\end{equation}
by construction of the linear system (see equations (\ref{eq:prob-intersection})
and (\ref{eq:linsyst})). Therefore, every relation of linear dependence
between the rows of $M$ holds also between the rows of $M|C\,,$
as required.

Notice also that the implication goes both ways: the linear system
has a solution \emph{only if} the set of probabilities $p(ab|xy)$
is no-signaling. The coefficient matrix $M$ incorporates the no-signaling
conditions by construction of the states $\omega_{a_{0}a_{1}b_{0}b_{1}}$
with probabilities $\mathsf{P}_{a_{0}a_{1}b_{0}b_{1}}\,.$ Therefore,
if these conditions do not hold for the independent terms $p(ab|xy)\,,$
then the rank of $M|C$ must be larger than that of $M\,,$ as $M|C$
contains more linearly independent rows than $M.$ 
\end{proof}

\section*{Supplementary Note 3 - Classical agreement in no-signalling boxes}

We now state and prove the classical agreement theorem in the no-signaling language, i.e. for local boxes. We restrict to boxes of two inputs and two outputs since, by Theorem \ref{thm:generalization-app}, any larger box exhibiting disagreement can be reduced to a 2-input 2-output box that also exhibits disagreement, while preserving its locality properties.

Recall that a distribution  $\{p(ab|xy):(a,b,x,y)\in\mathcal{A} \times \mathcal{B}\times\mathcal{X}\times\mathcal{Y}\}$  is local if, for each \sloppy $a,b,x,y,$
\begin{equation} \label{eq:plocal-app}
    p(ab|xy)=\sum_{\lambda} p_{\lambda}\, p_A(a|x\lambda)\,p_B(b|y\lambda),
\end{equation}
for some distributions  $\{p_{\lambda}:\lambda\in\Lambda\}$,  $\{p_A(a|x\lambda):(a,x,\lambda)\in\mathcal{A}\times\mathcal{X}\times\Lambda\},\,\{p_B(b|y\lambda):(b,y,\lambda)\in\mathcal{B}\times\mathcal{Y}\times\Lambda\}$ \color{black} where $\Lambda$ is an index set.

\begin{thm}\label{cor:aumann-app}
	Suppose Alice and Bob share a local no-signaling box with underlying probability
	distribution $p$. Let $q_A,q_B \in [0,1]$, and let
	\begin{equation}
	\begin{aligned}
	    p(b=1|a=0,x=0,y=1) &= q_A,\\
	    p(a=1|b=0,x=1,y=0) &= q_B.
	\end{aligned}
	\end{equation}
	If $q_{A}$ and $q_{B}$ are common certainty between the observers,	then $q_{A}=q_{B}.$
\end{thm}

\begin{proof}
	By definition of $q_{A},q_{B},$ and using the fact that the shared
	distribution is local and hence satisfies equation (\ref{eq:plocal-app}),
	we have
	\begin{equation}
	\begin{aligned}q_{A}\sum_{\lambda}p_{\lambda}p_{A}(0|0\lambda) & =\sum_{\lambda}p_{\lambda}p_{A}(0|0\lambda)p_{B}(1|1\lambda)\\
	q_{B}\sum_{\lambda}p_{\lambda}p_{B}(0|0\lambda) & =\sum_{\lambda}p_{\lambda}p_{A}(1|1\lambda)p_{B}(0|0\lambda)\,.
	\end{aligned}
	\label{eq:qAqBlocal}
	\end{equation}
	From Theorem \ref{thm:2in2outcc-app} we have that, if $1\in \alpha_{n}$
	or $1\in \beta_{n}$ for all $n\in\mathbb{N}$ then there is no common certainty of disagreement
	for any no-signaling distribution, and these encompass local distributions.
	Hence there only remains to prove the claim for $1\not\in \alpha_{n}$
	and $1\not\in \beta_{n},$ for some $n\in\mathbb{N}.$ This implies that
	\begin{equation}
	\begin{aligned}p(b=0|a=0,x=0,y=0) & =1\\
	p(a=0|b=0,x=0,y=0) & =1
	\end{aligned}
	\end{equation}
	and hence
	\begin{equation}
	\begin{aligned}\sum_{\lambda}p_{\lambda}p_{A}(0|0\lambda) & =\sum_{\lambda}p_{\lambda}p_{A}(0|0\lambda)p_{B}(0|0\lambda)\\
	\sum_{\lambda}p_{\lambda}p_{B}(0|0\lambda) & =\sum_{\lambda}p_{\lambda}p_{A}(0|0\lambda)p_{B}(0|0\lambda),
	\end{aligned}
	\end{equation}
	which implies, on the one hand, that
	\begin{equation}
	\sum_{\lambda}p_{\lambda}p_{A}(0|0\lambda)=\sum_{\lambda}p_{\lambda}p_{B}(0|0\lambda)\label{eq:00equal}
	\end{equation}
	and, on the other, that
	\begin{equation}
	\sum_{\lambda}p_{\lambda}p_{A}(0|0\lambda)p_{B}(1|0\lambda)=\sum_{\lambda}p_{\lambda}p_{A}(1|0\lambda)p_{B}(0|0\lambda)=0,
	\end{equation}
	that is,
	\begin{equation}
	p_{A}(0|0\lambda)p_{B}(1|0\lambda)=p_{A}(1|0\lambda)p_{B}(0|0\lambda)=0\label{eq:cc-local-app-app}
	\end{equation}
	for all $\lambda.$ Therefore, there remains to prove only that
	\begin{equation}
	\sum_{\lambda}p_{\lambda}p_{A}(0|0\lambda)p_{B}(1|1\lambda)=\sum_{\lambda}p_{\lambda}p_{A}(1|1\lambda)p_{B}(0|0\lambda).\label{eq:cc-local-apptoprove}
	\end{equation}
	Because the outputs for inputs $x=1,y=1$ are perfectly correlated,
	we have
	\begin{equation}
	p_{A}(0|1\lambda)p_{B}(1|1\lambda)=p_{A}(1|1\lambda)p_{B}(0|1\lambda)=0\label{eq:perfcorrlocal}
	\end{equation}
	for all $\lambda$ and, since $p_{A}(0|1\lambda)+p_{A}(1|1\lambda)=1$
	and similarly for $p_{B},$ this implies 
	\begin{equation}
	p_{A}(1|1\lambda)=p_{B}(1|1\lambda).\label{eq:11equal}
	\end{equation}
	Then we can prove (\ref{eq:cc-local-apptoprove}) by simple manipulations of the probability distributions of each party:
	\begin{equation}
	\begin{aligned}\sum_{\lambda}p_{\lambda}p_{A}(0|0\lambda)p_{B}(1|1\lambda) & =\sum_{\lambda} p_{\lambda}p_{A}(0|0\lambda)p_{A}(1|1\lambda)\left(p_{B}(0|0\lambda)+p_{B}(1|0\lambda)\right)\\
	& =\sum_{\lambda}p_{\lambda}p_{A}(0|0\lambda)p_{A}(1|1\lambda)p_{B}(0|0\lambda)\\
	& =\sum_{\lambda}p_{\lambda}\left(p_{A}(0|0\lambda)+p_{A}(1|0\lambda)\right)p_{A}(1|1\lambda)p_{B}(0|0\lambda)\\
	& =\sum_{\lambda}p_{\lambda}\left(p_{A}(0|1\lambda)+p_{A}(1|1\lambda)\right)p_{A}(1|1\lambda)p_{B}(0|0\lambda)\\
	& =\sum_{\lambda}p_{\lambda}p_{A}(1|1\lambda)p_{A}(1|1\lambda)p_{B}(0|0\lambda)\\
	& =\sum_{\lambda}p_{\lambda}p_{A}(1|1\lambda)p_{B}(0|0\lambda)
	\end{aligned}
	\end{equation}
	where we have used the fact that $\sum_{b\in\mathcal{B}}p_{B}(b|y\lambda)=1$
	for all $y,\lambda$ in the first equality, (\ref{eq:cc-local-app-app}) in the second
	and third, $\sum_{a\in\mathcal{A}}p_{A}(a|x\lambda)=1$ for all $x,\lambda$
	in the fourth, (\ref{eq:perfcorrlocal}) again in the fifth, and $p_{A}(1|1\lambda)^{2}=p_{A}(1|1\lambda)$
	for all $\lambda$ (since $p_{A}(a|x\lambda)$ is either 1 or 0 for
	every $a,x,\lambda$) in the last.
\end{proof}

\section*{Supplementary Note 4 - Proof of Theorem 3}

In the proof of Theorem \ref{thm:2in2outcc-app}, we will make use of the following Lemma:

\vspace{2mm}
\noindent\textit{Lemma A2.}{
Consider a no-signaling box of 2 inputs and 2 outputs. Then, $\alpha_0=\{0,1\}$ if and only if $q_A=p(b=1|y=1).$ Analogously, $\beta_0=\{0,1\}$ if and only if $q_B=p(a=1|x=1).$
}

\begin{proof}
We start by proving the direct implication. By hypothesis, 
\begin{align*}
q_{A} & =p(b=1|a=0,x=0,y=1)=\frac{p(01|01)}{p(a=0|x=0)}\\
 & =p(b=1|a=1,x=0,y=1)=\frac{p(11|01)}{p(a=1|x=0)}.
\end{align*}
But now, we can write 
\begin{align*}
p(b=1|y=1)=p(01|01)+p(11|01)=p(a=0|x=0)q_{A}+p(a=1|x=0)q_{A}=q_{A}.
\end{align*}
The reverse implication is trivial. The analogous statement can be proved
by interchanging the roles of Alice and Bob. 
\end{proof}

\begin{thm} \label{thm:2in2outcc-app}
A two-input two-output no-signaling box gives rise to common certainty of disagreement if and only if it takes the form of Supplementary Table \ref{tab:nsccd-app}.

\begin{table}[h]
\centering
\begin{tabular}{|c|c|c|c|c|}
\hline 
\label{isabox-app}$xy\backslash ab$  & 00  & 01  & 10  & 11\tabularnewline
\hline 
\hline 
00  & $r$  & 0  & 0  & $1-r$ \tabularnewline
\hline 
01  & $r-s$  & $s$  & $-r+t+s$  & $1-t-s$\tabularnewline
\hline 
10  & $t-u$  & $u$  & $r-t+u$  & $1-r-u$\tabularnewline
\hline 
11  & $t$  & 0  & 0  & $1-t$\tabularnewline
\hline 
\end{tabular}
\caption{Parametrization of two-input two-output no-signaling boxes with common certainty of disagreement. Here, $r,s,t,u  \in [0,1]$ are such that all the entries of the box are non-negative, $r>0$, and $s-u\neq r-t.$}
\label{tab:nsccd-app}
\end{table}
\end{thm}

\begin{proof}

We first prove that common certainty of disagreement imposes the
claimed structure for the no-signaling box. Therefore, we assume common certainty of disagreement, i.e., 
\begin{equation}\label{hyp1}
(0,0,0,0) \in A_n \cap B_n \qquad \forall n\in \mathbb{N}
\end{equation}
and
\begin{equation}
    q_A\neq q_B.
\end{equation}

%In particular, we also assume that Alice and Bob input $x=y=0$ and obtain
%$a=b=0.$ This implies
%\begin{equation}
%p(00|00)>0.
%\end{equation}

We split the proof into three cases
based on the contents of the sets $A_{n},B_{n}$: 
\begin{casenv}
	\item $1\notin \alpha_n,\, 1\notin \beta_{n}$ for some $n$.\footnote{This need not happen at the same stage, i.e., possibly $1\notin \alpha_{m}$,
		for some $m<n$. However in this case, since the sets are nonempty
		by assumption, we have $\alpha_{n}=\alpha_{m}.$ } From common certainty of disagreement (equation \eqref{hyp1}), we have that
	\begin{align*}
	p(B_{n}|a=0,x=0,y=0)=1,\qquad p(A_{n}|b=0,x=0,y=0)=1,
	\end{align*}
	which, together with $1\notin \alpha_n,\, 1\notin \beta_{n}$, translates into: 
	\begin{align*}
	p(01|00)=0,\qquad p(10|00)=0.
	\end{align*}
 For $q_A,q_B$ to be well-defined, $p(a=0|x=0)$ and $p(b=0|y=0)$ must be non-zero, therefore we must have $p(00|00)>0$.\color{black}
%We also assumed that the observers in fact obtained outputs $a=0,b=0$ on inputs $x=0,y=0$, so we must have $p(00|00)>0$.
The rest of the table is determined by no-signaling constraints in terms of
parameters $r$, $s$, $t$ and $u$. Given the box in the statement
of the theorem, $q_{A}\ne q_{B}$ if and only if $s-u\ne r-t$, 
which concludes the proof of this case. 
\item $\alpha_n=\lbrace0,1\rbrace$, for all $n\in\mathbb{N}$
while $1\notin \beta_{m}$ for some $m$. We show that this
case implies $q_{A}=q_{B}$, so it contradicts common
certainty of disagreement. Indeed, the definition of $\alpha_{m+1}$ enforces
the conditions: 
\begin{align*}
p(b=0|a=0,x=0,y=0)=1=p(b=0|a=1,x=0,y=0).
\end{align*}
This implies 
\begin{align*}
 0= \color{black} p(b=1|a=0,x=0,y=0) & =\frac{p(01|00)}{p(a=0|x=0)}\quad\Rightarrow\quad p(01|00)=0,\\
 0= \color{black} p(b=1|a=1,x=0,y=0) & =\frac{p(11|00)}{p(a=1|x=0)}\quad\Rightarrow\quad p(11|00)=0.
\end{align*}
Adding no-signaling conditions to these last equations, we also obtain
\begin{equation}
0=p(b=1|y=0)=p(01|10)+p(11|10),
\end{equation}
and so
\begin{equation}\label{ccd-zeroprob}
p(01|10)=0=p(11|10)
\end{equation}
and
\begin{equation}\label{ccd-oneprob}
    p(b=0|y=0)=1.
\end{equation}
This allows us to identify $q_{B}$ with $p(a=1|x=1)$, since
\begin{align*}
q_{B}&=  p(a=1|b=0,x=1,y=0)\\
&=\frac{p(10|10)}{p(b=0|y=0)}\\
&=p(10|10)\\
 & =p(a=1|x=1)-p(11|10)\\
 &=p(a=1|x=1),
\end{align*}
where the third and last equalities follow from equations \eqref{ccd-oneprob} and \eqref{ccd-zeroprob} respectively.
Now, taking into account Lemma A2 and perfect
correlations, we have
\[
q_{A}=p(b=1|y=1)=p(a=1|x=1),
\]
which shows that $q_{A}=q_{B}$, as mentioned above. 
\item $\alpha_n=\lbrace0,1\rbrace$, $\beta_n=\lbrace0,1\rbrace$ for all $n\in\mathbb{N}$. We now show that this case also implies
$q_{A}=q_{B}$, contradicting common certainty of disagreement.
Using Lemma A2 we have 
\[
q_{B}=p(a=1|x=1)\qquad\text{ as well as }\qquad q_{A}=p(b=1|y=1).
\]
Now, perfect correlations impose that $p(a=1|x=1)=p(b=1|y=1)$, that
is, $q_{A}=q_{B}$. 
\end{casenv}
Next, we prove the converse implication of the theorem. We show
that any no-signaling box of the above form must exhibit common certainty of
disagreement. Since $s-u\ne r-t$,  the probabilities Alice and Bob assign to their events of interest are different\color{black}:
\begin{equation}
\begin{aligned}q_{A}:\color{black}=p(b=1|a=0,x=0,y=1) & =s/r,\\
q_{B}:\color{black}=p(a=1|b=0,x=1,y=0) & =(r-t+u)/r\,.
\end{aligned}
\label{eq:qAqB}
\end{equation}

In the case that $1\not\in \alpha_{0},\:1\not\in \beta_{0},$
we also have that $\alpha_1=\alpha_0$ and $\beta_1=\beta_0\,,$ and common
certainty of disagreement follows, because $(0,0,0,0)$ is in $A_n\cap B_n$ for all $n$.

If the parameters are such that 
\begin{equation}
\frac{1-t-s}{1-r}=\frac{s}{r},
\end{equation}
but 
\begin{equation}
\frac{1-r-u}{1-r}\neq\frac{r-t+u}{r},
\end{equation}
then 
\begin{equation}
p(b=1|a=1,x=0,y=1)=q_{A},
\end{equation}
as well, but 
\begin{equation}
p(a=1|b=1,x=1,y=0)\neq q_{B},
\end{equation}
and so $1\in \alpha_{0},\:1\not\in \beta_{0}.$ Since we have
\begin{equation}
p(b=0|a=0,x=0,y=0)=1,
\end{equation}
we find $(0,0,0,0)\in A_{1}$,\footnote{Note $(1,0,0,0)\not\in A_{1},$ though this does not affect the present
proof.} and hence all $A_{n}$ still contain $(0,0,0,0),$ yielding
common certainty of disagreement.

Symmetric reasoning covers the case $1\not\in \alpha_{0},\:1\in \beta_{0}$, and only the case where $\alpha_0=\{0,1\}$, $\beta_0=\{0,1\}$ remains. This happens when 
\begin{equation}
    \begin{aligned}
    p(b=1|a=1,x=0,y=1)&=p(b=1|a=0,x=0,y=1),\\
    p(a=1|b=0,x=1,y=0)&=p(a=1|b=1,x=1,y=0)
    \end{aligned}
\end{equation}
which, in terms of the parameters, is equivalent to
\begin{align}
    \frac{1-t-s}{1-r}&=\frac{s}{r},\label{cond1}\\
    \frac{1-r-u}{1-r}&=\frac{r-t+u}{r}\label{cond2}.
\end{align}
However, these two conditions are satisfied simultaneously only when
$s-u=r-t$, as we now show. From Equation \eqref{cond1} we get 
\[
s=r(1-t),
\]
while from Equation \eqref{cond2} we obtain 
\[
u=t(1-r).
\]
This means that if Equations \eqref{cond1} and \eqref{cond2} are both satisfied,
then 
\[
s-u=r(1-t)-t(1-r)=r-t,
\]
 which contradicts the assumption that $s-u\neq r-t$. \color{black}
\end{proof}

\begin{thm} \label{thm:ccd-not-quantum}
No two-input two-output quantum box can give rise to common certainty of disagreement.
\end{thm}

\begin{proof}
In order to give rise to common certainty of disagreement, the probability
distribution that the state and measurements generate must be of the
form of Table \ref{tab:nsccd-app}. Theorem 1 in Tsirelson's seminal paper \cite{cirelson_quantum_1980} implies that, if there is a quantum realization
of the box, then there exist real, unit vectors 
\begin{equation}
\ket{w_{x}},\ket{v_{y}}
\end{equation}
such that the correlations 
\begin{equation}
c_{xy}:=p(a=b|xy)-p(a\neq b|xy)
\end{equation}
satisfy 
\begin{equation}
c_{xy}=\bra{w_{x}}\ket{v_{y}}
\end{equation}
for each $x,y\in\{0,1\}.$ \color{black} For the box in Theorem \ref{thm:2in2outcc-app},
this means, in particular, that 
\begin{equation}
\begin{aligned}\bra{w_{0}}\ket{v_{0}} & =1\,,\\
\bra{w_{1}}\ket{v_{1}} & =1\,,
\end{aligned}
\end{equation}
and, since the vectors have unit norm, this implies that 
\begin{equation}
\begin{aligned}\ket{w_{0}} & =\ket{v_{0}}\,,\\
\ket{w_{1}} & =\ket{v_{1}}\,.
\end{aligned}
\end{equation}
Then, we are left with 
\begin{equation}
\begin{aligned}c_{01} & =\bra{w_{0}}\ket{w_{1}}\,,\\
c_{10} & =\bra{w_{1}}\ket{w_{0}}\,.
\end{aligned}
\end{equation}
Since the vectors are real, we find 
\begin{equation}
c_{01}=c_{10}\,,
\end{equation}
but this implies that 
\begin{equation}
s-u=r-t,
\end{equation}
which implies that $q_A=q_B$ and, hence, impedes disagreement. 
\end{proof}

\section*{Supplementary Note 5 - Proof of Theorem 5}

\begin{thm} \label{thm:generalization-app}
No quantum box can give rise to common certainty of disagreement.
\end{thm}

\begin{proof}
We show that any no-signaling box with common certainty of disagreement induces a 2-input 2-output no-signaling box with the same property. Thus, if there existed a quantum system that could generate the bigger box, it could also generate the smaller box. Then, Theorem \ref{thm:2in2outcc-app} implies that no quantum box can give rise to common certainty of disagreement.

We define a mapping from a distribution  $\left\{ p(ab|xy):a\in\mathcal{A},b\in\mathcal{B},x\in\mathcal{X},y\in\mathcal{Y}\right\}$ \color{black}
to an `effective' distribution  $\left\{ \tilde{p}(\tilde{a}\tilde{b}|\tilde{x}\tilde{y}):\tilde{a},\tilde{b},\tilde{x},\tilde{y}\in\left\{ 0,1\right\} \right\}$ \color{black}
such that the following conditions hold: 
\begin{enumerate}
\item if $\left\{ p(ab|xy)\right\} $ is quantum, then
so is $\left\{ \tilde{p}(\tilde{a}\tilde{b}|\tilde{x}\tilde{y})\right\} \,,$\label{enu:ptilde-NS} 
\item if $\left\{ p(ab|xy)\right\} $ satisfies common certainty of disagreement,
then so does $\left\{ \tilde{p}(\tilde{a}\tilde{b}|\tilde{x}\tilde{y})\right\} \,.$\label{enu:ptilde-ccd} 
\end{enumerate}
First, notice that the number of inputs can be reduced to 2 without
loss of generality, as common certainty of disagreement is always
defined to be \emph{at }an event (wlog, $(0,0,0,0)$) \emph{about}
another event (wlog, $(1,1,1,1)$). One can associate the inputs $x=0,y=0$
with $\tilde{x}=0,\tilde{y}=0$, respectively, and $x=1,y=1$ with
$\tilde{x}=1,\tilde{y}=1$ respectively, and ignore all other possible
inputs in $\mathcal{X},\mathcal{Y}.$%
\begin{comment}
\textbf{{[}WELL, we could consider the observers not knowing each other's
inputs, in which case it might make sense to have more inputs. Is
this too general?{]}} 
\end{comment}
{} The outputs, instead, must be grouped according to whether or not
they belong in the sets $A_{n},B_{n}$ (for input 0) and whether or
not they correspond to the event obtaining, i.e. whether or not they
are equal to 1 (for input 1).

Since $p$ satisfies common certainty of disagreement,
we know that $(0,0,0,0)\in A_{n}\cap B_{n}.$ Moreover, by the definitions
of the sets $\alpha_{n},\beta_{n}$ (and since we only consider finite sets $\mathcal{A},\mathcal{B},\mathcal{X},\mathcal{Y}$) there exists an $N\in\mathbb{N}$
such that $\alpha_{n}=\alpha_{N}$ and $\beta_{n}=\beta_{N}$ for
all $n\geq N\,.$ Take such $N,$ and define the following indicator
functions: 
\begin{equation}
\begin{aligned}\chi_{0|0}^{\alpha}(a) & =\begin{cases}
0 & a\not\in\alpha_{N}\\
1 & a\in\alpha_{N}
\end{cases}\\
\chi_{0|0}^{\beta}(b) & =\begin{cases}
0 & b\not\in\beta_{N}\\
1 & b\in\beta_{N}
\end{cases}\\
\chi_{0|1}^{\alpha}(c)=\chi_{0|1}^{\beta}(c) & =\begin{cases}
0 & c=1\\
1 & c\neq1
\end{cases}
\end{aligned}
\end{equation}
(where $c$ stands for output $a,b$ for Alice and Bob, respectively),
with 
\begin{equation}
\begin{aligned}
\chi_{1|x}^{\alpha}(a) & =1-\chi_{0|x}^{\alpha}(a)\,\\
\chi_{1|y}^{\beta}(b) & =1-\chi_{0|y}^{\beta}(b)
\end{aligned}
\end{equation}
for each $a,b,x,y$. Then, the mapping from $p$ to $\tilde{p}$
is defined as follows: 
\begin{equation}
\tilde{p}(\tilde{a}\tilde{b}|\tilde{x}\tilde{y})=\sum_{a,b}\delta_{x,\tilde{x}}\delta_{y,\tilde{y}}\chi_{\tilde{a}|x}^{\alpha}(a)\chi_{\tilde{b}|y}^{\beta}(b)p(ab|xy)\label{eq:ptoptilde}
\end{equation}
where 
\begin{equation}
\delta_{s,t}=\begin{cases}
0 & s\neq t\\
1 & s=t\,.
\end{cases}
\end{equation}
\begin{comment}
\begin{equation}
\begin{aligned}\tilde{p}(00|00) & =\sum_{\substack{a\in\alpha_{N}\\
b\in\beta_{N}
}
}p(ab|00)\\
\tilde{p}(01|00) & =\sum_{\substack{a\in\alpha_{N}\\
b\not\in\beta_{N}
}
}p(ab|00)\\
\tilde{p}(10|00) & =\sum_{\substack{a\not\in\alpha_{N}\\
b\in\beta_{N}
}
}p(ab|00)\\
\tilde{p}(11|00) & =\sum_{\substack{a\not\in\alpha_{N}\\
b\not\in\beta_{N}
}
}p(ab|00)\\
\tilde{p}(00|01) & =\sum_{\substack{a\in\alpha_{N}\\
b\neq1
}
}p(ab|01)\\
\tilde{p}(01|01) & =\sum_{\substack{a\in\alpha_{N}}
}p(a1|01)\\
\tilde{p}(10|01) & =\sum_{\substack{a\not\in\alpha_{N}\\
b\neq1
}
}p(ab|01)\\
\tilde{p}(11|01) & =\sum_{\substack{a\not\in\alpha_{N}}
}p(a1|01)\\
\tilde{p}(00|10) & =\sum_{\substack{a\neq1\\
b\in\beta_{N}
}
}p(ab|10)\\
\tilde{p}(01|10) & =\sum_{\substack{a\neq1\\
b\not\in\beta_{N}
}
}p(ab|10)\\
\tilde{p}(10|10) & =\sum_{\substack{b\in\beta_{N}}
}p(1b|10)\\
\tilde{p}(11|10) & =\sum_{\substack{b\not\in\beta_{N}}
}p(1b|10)\\
\tilde{p}(00|11) & =\sum_{\substack{a\neq1\\
b\neq1
}
}p(ab|11)\\
\tilde{p}(01|11) & =\sum_{\substack{a\neq1}
}p(a1|11)\\
\tilde{p}(10|11) & =\sum_{\substack{b\neq1}
}p(1b|11)\\
\tilde{p}(11|11) & =p(11|11)\,.
\end{aligned}
\end{equation}
\end{comment}

\begin{figure}
	\centering
	\includegraphics[width=0.5\textwidth]{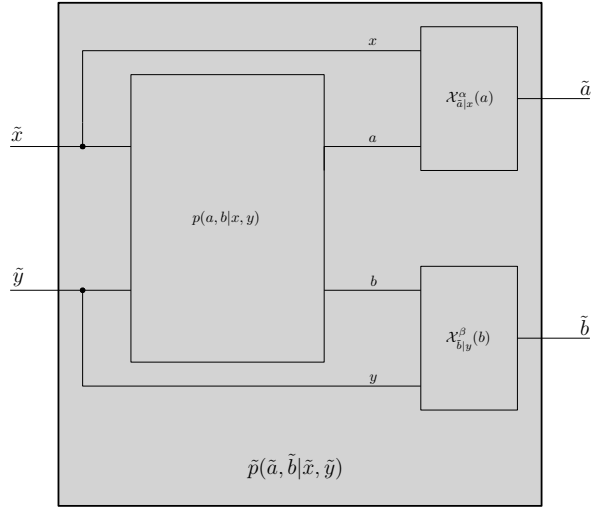}
	\caption{A diagrammatic representation of the construction of $\tilde p $.}
	\label{fig1}
\end{figure}

We note that the distribution $\tilde{p}$ is     obtained from $ p $ by means of local pre- and post-processing (see Supplementary Figure \ref{fig1} for a pictorial representation of the transformation). The classicality and locality of this transformation ensure that  $\tilde{p}$ is a well defined non-signalling distribution that is no more non-local than $p$. In particular, condition 1 above is satisfied.

To check condition \ref{enu:ptilde-ccd},  i.e. that if ${p(ab|xy)}$ satisfies common certainty of disagreement, then so does ${\tilde{p}(ab|xy)}$, \color{black} let $N$ be as in the definition
of the map (\ref{eq:ptoptilde}) and let $a\in\alpha_{N}.$ Then,
by definition of the set $\alpha_{N+1},$ we have

\begin{equation}
p(\beta_{N}|a,x=0,y=0)=1
\end{equation}
and, therefore, 
\begin{equation}
\frac{\sum_{b\in\beta_{N}}p(ab|00)}{\sum_{b\in\mathcal{B}}p(ab|00)}=1\,,
\end{equation}
which entails 
\begin{equation}
\sum_{b\not\in\beta_{N}}p(ab|00)=0\,.
\end{equation}
Summing over $a\in\alpha_{N}\,,$ we get 
\begin{equation}
\sum_{\substack{a\in\alpha_{N}\\
b\not\in\beta_{N}
}
}p(ab|00)=\tilde{p}(01|00)=0\,.
\end{equation}
Similarly, we find $\tilde{p}(10|00)=0\,.$ Since $p$ satisfies common
certainty of disagreement, its outputs on input $x=1,y=1$ must be
perfectly correlated. That is, $p(ab|11)=0$ if $a\neq b\,.$ Hence,
\begin{equation}
\tilde{p}(01|11)=\sum_{a\neq1}p(a1|11)=0
\end{equation}
and similarly for $\tilde{p}(10|11)\,.$ So far, the no-signaling
box corresponding to $\tilde{p}$ has two zeros in the first row and
another two in the last. Using normalization and no-signaling conditions
to fill in the rest of the table, we find it is of the form of the
no-signaling box in Theorem \ref{thm:2in2outcc-app}. There remains to check
for disagreement, i.e. that if 
\begin{equation}
q_{A}=p(b=1|a=0,x=0,y=1)\neq p(a=1|b=0,x=1,y=0)=q_{B}
\end{equation}
then 
\begin{equation}
\tilde{p}(\tilde{b}=1|\tilde{a},\tilde{x}=0,\tilde{y}=1)\neq\tilde{p}(\tilde{a}=1|\tilde{b},\tilde{x}=1,\tilde{y}=0)\,.
\end{equation}
Since $\alpha_{N}\subseteq\alpha_{0}$ and $\beta_{N}\subseteq\beta_{0},$
$p(b=1|a^{*},x=0,y=1)\neq p(a=1|b^{*},x=1,y=0)$ holds in particular
for all $a^{*}\in\alpha_{N},b^{*}\in\beta_{N}.$ This means that,
for $a^{*}\in\alpha_{N},b^{*}\in\beta_{N},$ 
\begin{equation}
\frac{p(a^{*}1|01)}{\sum_{b\in\mathcal{B}}p(a^{*}b|01)}\neq\frac{p(1b^{*}|10)}{\sum_{a\in\mathcal{A}}p(ab^{*}|10)}
\end{equation}
and so 
\begin{equation}
p(a^{*}1|01)\sum_{a\in\mathcal{A}}p(ab^{*}|10)\neq p(1b^{*}|10)\sum_{b\in\mathcal{B}}p(a^{*}b|01)\,.
\end{equation}
Then, we can sum over $\alpha_{N}$ and $\beta_{N}$ on both sides
to find 
\begin{equation}
\sum_{a^{*}\in\alpha_{N}}p(a^{*}1|01)\sum_{\substack{a\in\mathcal{A}\\
b^{*}\in\beta_{N}
}
}p(ab^{*}|10)\neq\sum_{b^{*}\in\beta_{N}}p(1b^{*}|10)\sum_{\substack{a^{*}\in\alpha_{N}\\
b\in\mathcal{B}
}
}p(a^{*}b|01)\,.
\end{equation}
But in terms of $\tilde{p},$ this corresponds to 
\begin{equation}
\tilde{p}(01|01)\sum_{\tilde{a}\in\{0,1\}}\tilde{p}(\tilde{a}0|10)\neq\tilde{p}(10|10)\sum_{\tilde{b}\in\{0,1\}}\tilde{p}(0\tilde{b}|01)
\end{equation}
which implies 
\begin{equation}
\tilde{p}(\tilde{b}=1|\tilde{a}=0,\tilde{x}=0,\tilde{y}=1)\neq\tilde{p}(\tilde{a}=1|\tilde{b}=0,\tilde{x}=1,\tilde{y}=0)
\end{equation}
and hence the disagreement occurs for the $\tilde{p}$ distribution
as well, which proves the result.

Notice that the sets $\tilde{\alpha}_{0},\tilde{\beta}_{0}$ in the
distribution $\tilde{p}$ (defined analogously to $\alpha_{0},\beta_{0}$
in the distribution $p$) will correspond to outputs $\tilde{a},\tilde{b}=0$,
respectively. This is to be expected, as the map $p\rightarrow\tilde{p}$
gives rise to a no-signaling box of the form of the one in Theorem \ref{thm:2in2outcc-app},
where the sets $\tilde{\alpha}_{0},\tilde{\beta}_{0}$ contain a single
element each. (In effect, this means we are ignoring the outputs
$a^{*}\in\alpha_{0}\backslash\alpha_{N}$ and $b^{*}\in\beta_{0}\backslash\beta_{N}$,
but those outputs lead to disagreement but not to common
certainty of it, so they can be safely discarded.)

\end{proof}

\section*{Supplementary Note 6 - Singular disagreement}

\begin{thm}
\label{thm:loc-sd}There is no local two-input two-output box that gives rise to singular disagreement. 
\end{thm}

\begin{proof}
Assume Alice and Bob input $x=y=0$ and obtain
$a=b=0.$ This implies
\begin{equation}
p(00|00)>0.\label{eq:singdis00>0-1}
\end{equation}
Alice assigns 
\begin{equation}
p(b=1|a=0,x=0,y=1)=1,\label{eq:singdisAlice1-1}
\end{equation}
and Bob assigns 
\begin{equation}
p(a=1|b=0,x=1,y=0)=0\,.\label{eq:singdisBob0-1}
\end{equation}
Further, the outputs for input $(x,y)=(1,1)$ are perfectly correlated,
so, in particular, 
\begin{equation}
p(01|11)=0.\label{eq:singdisPerfcorr-1}
\end{equation}
Equations \eqref{eq:singdisAlice1-1} and \eqref{eq:singdisBob0-1} imply, respectively,
\begin{equation}
    p(00|01)=0 \textnormal{ and } p(10|10)=0.\label{eq:singdisAliceBob}
\end{equation}
However, equations \eqref{eq:singdis00>0-1}, \eqref{eq:singdisPerfcorr-1} and \eqref{eq:singdisAliceBob} make up a form of Hardy's paradox \cite{hardy_quantum_1992}, which is known not to hold for local distributions.
\end{proof}

\begin{thm} \label{thm:singdis-app}
A two-input two-output no-signaling box gives rise to singular disagreement if and only if it takes the form of Supplementary Table \ref{tab:nssd}.
\begin{table}[ht]
\centering
\begin{tabular}{|c|c|c|c|c|}
\hline 
$xy\backslash ab$  & 00  & 01  & 10  & 11\tabularnewline
\hline 
\hline 
00  & $s$  & $t$  & $1-s-u-t$  & $u$\tabularnewline
\hline 
01  & 0  & $s+t$  & $r$  & $1-s-t-r$\tabularnewline
\hline 
10  & $1-u-t$  & $u+t+r-1$  & 0  & $1-r$\tabularnewline
\hline 
11  & $r$  & 0  & 0  & $1-r$\tabularnewline
\hline 
\end{tabular}
\caption{Parametrization of two-input two-output no-signaling boxes with singular disagreement. Here, $r,\,s,\, t,\, u,\, \in [0,1]$ are such that all the entries of the box are non-negative, $s>0$, and $s + t \neq 0$ and $u+t\neq 1$.}
\label{tab:nssd}
\end{table}

\end{thm}

\begin{proof}
First, we show that singular disagreement implies that the no-signaling box must
be of the above form.  By construction, the inputs $x=y=1$ have perfectly correlated outputs, so that
\begin{equation}
p(01|11)=p(10|11)=0\,.
\end{equation}
Also, singular disagreement requires

\begin{eqnarray}
 & p(b=1|a=0,x=0,y=1) & =1,\label{eq:sd-Alice1}\\
 & p(a=1|b=0,x=1,y=0) & =0.\label{eq:sd-Bob0}
\end{eqnarray}
Equation (\ref{eq:sd-Alice1}) implies that $p(00|01)=0$ and $p(01|01)\neq 0$, while Equation (\ref{eq:sd-Bob0}) implies that $p(10|10)=0$ and $p(00|10)\neq 0$.
The rest of the entries follow from normalization and no-signaling conditions.  The condition $s>0$ ensures that $p(00|00)>0$, as per the input and output that the observers in fact obtained.
Therefore, any two-input two-output no-signaling box that gives rise to singular disagreement must be of the above form.

Proving the converse is straightforward, as it suffices to check that
Equations (\ref{eq:sd-Alice1}) and (\ref{eq:sd-Bob0}) are satisfied
for the parameters of the box. 
\end{proof}

\begin{thm}
No two-input two-output quantum box  can  give  rise  to  singular disagreement.
\end{thm}
\begin{proof}
 Ref. \cite{rai_geometry_2019} guarantees that boxes of Theorem \ref{thm:singdis-app} are either local or postquantum. \color{black}
This can be seen by observing that the mapping
\begin{equation}
x\mapsto x\oplus1\,,
\end{equation}
which is a symmetry of the box, makes all four 0's lie in entries $p(ab|xy)$ such that $a \oplus b \oplus 1=xy$. As stated in Sections III and V.B of  Ref.~\cite{rai_geometry_2019}, all boxes with four 0's in entries of the above form lie in quantum voids.
 In fact, the boxes in Theorem \ref{thm:singdis-app} cannot be local as we have seen that they contain a Hardy paradox, so that all such boxes are postquantum. \color{black}
\end{proof}

\begin{thm} \label{thm:generalization-app-sd}
No quantum box can give rise to singular disagreement.
\end{thm}

\begin{proof}
Like in Theorem \ref{thm:generalization-app}, we show that any no-signaling box with singular disagreement induces a 2-input 2-output no-signaling box with the same property, and rely on Theorem 8 to deduce that no quantum system can give rise to singular disagreement.

Analogously to Theorem \ref{thm:generalization-app}, to prove the Theorem for singular disagreement we
define a mapping from a distribution  $\left\{ p(ab|xy) :a\in\mathcal{A},b\in\mathcal{B},x\in\mathcal{X},y\in\mathcal{Y}\right\}$
\color{black} to an `effective' distribution  $\left\{ \tilde{p}(\tilde{a}\tilde{b}|\tilde{x}\tilde{y}):\tilde{a},\tilde{b},\tilde{x},\tilde{y}\in\left\{ 0,1\right\}\right\}$ \color{black}
such that the following conditions hold: 
\begin{enumerate}
\item if $\left\{ p(ab|xy)\right\} $ is quantum, then
so is $\left\{ \tilde{p}(\tilde{a}\tilde{b}|\tilde{x}\tilde{y})\right\} \,,$\label{enu:ptilde-NS-singdis} 
\item if $\left\{ p(ab|xy)\right\} $ satisfies singular disagreement, then
so does $\left\{ \tilde{p}(\tilde{a}\tilde{b}|\tilde{x}\tilde{y})\right\} \,.$\label{enu:ptilde-singdis} 
\end{enumerate}
Again, the number of inputs can be reduced to 2 without loss of generality\textbf{.}
To group the outputs, we notice that the sets $A_{0},B_{0}$ also
play a role in singular disagreement, as they group the outputs of
each party which lead them to assign their respective probabilities
to the event. Then, we group the outputs according to whether or not
they belong in the sets $\alpha_{0},\beta_{0}$ (for input 0) and
whether or not they correspond to the event obtaining, i.e. whether
or not they are equal to 1 (for input 1). We obtain the same mapping
(\ref{eq:ptoptilde}) as before, substituting $\alpha_{N}$ for $\alpha_{0}$
and $\beta_{N}$ for $\beta_{0}$. With this replacement, condition
\ref{enu:ptilde-NS-singdis} follows by the same proof as before.
To check condition \ref{enu:ptilde-singdis}, we know that, for all
$a^{*}\in\alpha_{0},$ 
\begin{equation}
p(b=1|a^{*},x=0,y=1)=1
\end{equation}
and so 
\begin{equation}
\begin{aligned}p(a^{*}1|01) & =\sum_{b\in\mathcal{B}}p(a^{*}b|01).\end{aligned}
\end{equation}
Summing over $a^{*}\in\alpha_{0}$ and rewriting the expression in
terms of $\tilde{p},$ we find 
\begin{equation}
\tilde{p}(01|01)=\sum_{\tilde{b}\in\{0,1\}}\tilde{p}(0\tilde{b}|01)
\end{equation}
which implies 
\begin{equation}
\tilde{p}(\tilde{b}=1|\tilde{a}=0,\tilde{x}=0,\tilde{y}=1)=1.
\end{equation}
Similarly, for all $b^{*}\in\beta_{0}$ we have 
\begin{equation}
p(a=1|b^{*},x=1,y=0)=0,
\end{equation}
hence 
\begin{equation}
p(1b^{*}|10)=0
\end{equation}
and so, by adding over $b^{*}\in\beta_{0}$ and mapping to $\tilde{p},$
we find 
\begin{equation}
\tilde{p}(10|10)=0
\end{equation}
as required. 
\end{proof}

%
% \bibliographystyle{plain}
%\bibliography{ADNCW}
%
%\end{document}

\end{document}